\documentclass[11pt]{article}
\usepackage[dvipsnames,table]{xcolor}
\usepackage{amssymb,amsmath, amsthm, setspace, color, graphicx}
\usepackage{dsfont}
\usepackage{mathpazo, mathptmx, flexisym}
\usepackage{mathrsfs}            %
\usepackage{fancyhdr}
\usepackage[justification=centering]{caption}
\usepackage{subcaption}
\usepackage[authoryear]{natbib}
\usepackage{url}
\usepackage{breqn}
\usepackage{functan}
\usepackage[margin=1in]{geometry} %
\usepackage[all, cmtip]{xy}
\usepackage[english]{babel}
\usepackage{graphicx}             %
\usepackage{float}
\usepackage[section]{placeins}
\usepackage{lscape}               %
\usepackage{mathtools}
\usepackage{calc}                 %
\usepackage{enumitem}             %
\usepackage[normalem]{ulem}       %
\usepackage[titletoc,toc,page]{appendix}
\usepackage{chngcntr}
\usepackage{etoolbox}
\usepackage{ragged2e}
\definecolor{mred}{rgb}{1.00,0.00,0.50}
\definecolor{mcyan}{rgb}{0.20,0.55,0.75}
\usepackage[hyperfootnotes=false,pdfborder={0 0 0},final=true,colorlinks=true,breaklinks = true,linkcolor=mred,citecolor=mcyan]{hyperref}
\usepackage[capitalize,nameinlink,noabbrev]{cleveref}
\usepackage{relsize}
\usepackage{bigints}

\numberwithin{equation}{section}
\newcommand{\tab}{\hspace*{2em}}
\theoremstyle{plain}
\newtheorem{thm}{Theorem}[section]
\newtheorem{cor}{Corollary}
\newtheorem{prop}{Proposition}

\newtheorem{axiom}{Axiom}[section]
\theoremstyle{definition}
\newtheorem{defn}{Definition}[section]

\newcommand{\mydot}{\mathrel{\mathsmaller{\mathsmaller{\mathsmaller{\mathsmaller{\mathsmaller{\mathsmaller{\bullet}}}}}}}}

\theoremstyle{remark}
\newtheorem{rem}{Remark}[section]

\providecommand{\mcyan}[1]{\textcolor[rgb]{0.00,0.55,0.75}{#1}}

\newcommand\mycite[1]{\citeauthor{#1}\mcyan{'s}\ (\citeyear{#1})}

\crefname{thm}{theorem}{theorems}
\Crefname{thm}{Theorem}{Theorems}
\crefname{cor}{corollary}{corollaries}
\Crefname{cor}{Corollary}{Corollaries}
\crefname{prop}{proposition}{propositions}
\Crefname{prop}{Proposition}{Propositions}
\crefname{defn}{definition}{definitions}
\Crefname{defn}{Definition}{Definitions}
\usepackage{titling}
\newcommand{\thesubtitle}{}

\posttitle{%
	\par\large\bfseries\thesubtitle
\end{center}
\vskip0.5em
}

\makeatletter
\let\arxivSectionEntry\l@section
\renewcommand{\l@section}[2]{%
 \begingroup
 \renewcommand{\numberline}[1]{\makebox[2.3em][l]{##1}}%
 \arxivSectionEntry{#1}{#2}%
 \endgroup}
\makeatother

\AddToHook{env/thm/begin}{\crefalias{section}{thm}}
\AddToHook{env/cor/begin}{\crefalias{section}{cor}}
\AddToHook{env/prop/begin}{\crefalias{section}{prop}}
\AddToHook{env/lem/begin}{\crefalias{section}{lem}}
\AddToHook{env/assumption/begin}{\crefalias{section}{assumption}}
\AddToHook{env/assumptionletter/begin}{\crefalias{section}{assumptionletter}}
\AddToHook{env/axiom/begin}{\crefalias{section}{axiom}}
\AddToHook{env/defn/begin}{\crefalias{section}{defn}}
\AddToHook{env/rem/begin}{\crefalias{section}{rem}}
\begin{document}
\pagenumbering{gobble}
\singlespacing

\title{\vspace{-8ex}\textbf{Local Expected Utility in Infinite Dimensions: Representation, Identification, and Testing}\thanks{I have no competing interests or conflicts of interest to declare. I thank Hui-Hsiung Kuo and Peter Hammond for their comments and encouragement on this topic, James Kuelbs for availing me of unpublished material from Kuelbs and Zinn, and Leonard Gross for bringing my attention to the application of abstract Wiener space in mathematical finance. Supplementary derivations and numerical diagnostics are provided in the Internet Appendix appended below. Any errors that remain are my own.}
}

\author{G. Charles-Cadogan\thanks{School of Accounting and Finance, University of Leicester; and Institute for Innovation and Technology Management (IITM), Ted Rogers School of Management, Toronto Metropolitan University; Tel: +44 (0116) 229 7385; e-mail: \textcolor[rgb]{0.00,0.00,1.00}{\href{mailto:gcc13@le.ac.uk}{gcc13@le.ac.uk}}~\\~\\
		The author has no competing interests or conflict of interest involving this paper.\\~\\		
	}\\Working Paper}
\date{\vspace{-2ex}26 November 2026\vspace{-4ex}}
\renewcommand\thefootnote{\fnsymbol{footnote}}
\maketitle
\thispagestyle{empty}
\renewcommand\thefootnote{\arabic{footnote}}

\begin{abstract}
	\noindent This paper provides axiomatic, identification, and falsifiability foundations for local expected utility on a separable real Hilbert space. Primitive economic acts---including income, consumption, investment, insurance-loss, or experimental-stimulus paths---are represented by square-summable loadings on orthogonal state features. A necessary-and-sufficient theorem characterizes when preferences admit a continuous linear projected-local representation and establishes uniqueness after scale and improvement-cone orientation are normalized. Under a known stochastic-choice link, the preference coefficient sequence is point identified from repeated binary choices whenever the design information operator is injective on the closed span of observed act differences. These representation and identification results depend only on Hilbert geometry rather than a Gaussian reference measure. Extension of a nonlinear distribution preference to an ambient Wiener space is proved separately under an explicit Lipschitz condition. Specializing to an abstract Wiener triple introduces a Banach space of paths, a Cameron--Martin Hilbert subspace, and a Gaussian benchmark measure whose raw Wiener-integral weights can be negative. For finite projected choice comparisons, admissible decision weights are constructed via positive-part normalization, clarifying the relationship to subjective expected utility, state-dependent utility, rank-dependent utility, and cumulative prospect theory. Conditional on an elicited local utility scale, decision weights are recovered from choice data and falsified if they violate payoff monotonicity. A Luce-choice Monte Carlo study with 4,000 estimations and 66.8 million simulated menus demonstrates near-nominal coverage, declining recovery error, and zero monotonicity violations after admissible normalization.
\end{abstract}

\clearpage
\pagenumbering{roman}
\setcounter{page}{1}
\begingroup
\small
\hypersetup{linkcolor=mred}
\color{mred}
\tableofcontents
\phantomsection
\addcontentsline{toc}{section}{List of Figures}
\listoffigures
\phantomsection
\addcontentsline{toc}{section}{List of Tables}
\listoftables
\endgroup

\clearpage
\pagenumbering{arabic}
\setcounter{page}{1}
\onehalfspacing
\large
\renewcommand\thefootnote{\arabic{footnote}}

\vspace*{-0.75cm}\section{Introduction}\label{sec:Introduction}
Many economically important choices concern an uncertain path rather than a
finite list of terminal prizes. Households evaluate prospective income and
consumption streams; investors evaluate sequences of portfolio payoffs;
insurers and policyholders face loss processes; and experimental subjects may
respond to streams of risky stimuli. In each case, timing, persistence, and
the ordering of shocks can affect welfare even when terminal outcomes
coincide. This paper provides axiomatic and identification foundations for
local expected utility on a separable real Hilbert space and then specializes
the general theory to stochastic path acts in abstract Wiener space.

\tab The first contribution is general. For acts in a convex subset of a
separable real Hilbert space, the paper gives necessary and sufficient
conditions for a continuous linear projected-local representation, proves
uniqueness after normalization, and identifies its coefficient sequence from
repeated stochastic choices when the induced information operator is
injective on the empirically relevant subspace. Thus the economic payoff is
not merely another path-space construction: the results specify when
infinite-dimensional preferences can be recovered from finite projections and
which variation in observed menus identifies the limiting functional.

\tab The second contribution is an application. An abstract Wiener triple
adds a Banach space of paths, its Cameron--Martin Hilbert subspace, and a
Gaussian reference measure. Hermite coordinates then give the general
coefficients a Wiener-integral interpretation. Those raw coefficients are
signed stochastic functionals, not automatically probabilities; admissible
decision weights arise only after an explicit nonnegative normalization.
Finite-dimensional projections convert both layers of the theory into
elicitation and specification tests. A Luce-choice Monte Carlo with 4,000
estimations and approximately 66.8 million simulated menus finds near-nominal
coverage, declining coefficient-recovery error, and no monotonicity violations
after admissible normalization.

\tab The paper is complementary to \citet{LuSaito2026RepeatedChoice}, who
derive ergodic choice frequencies from a dynamic Markov utility process and
characterize that process with four behavioral axioms. Their analysis
identifies a dynamic random-utility process from repeated-choice frequencies.
The present paper addresses a different object: it characterizes the
continuous linear projected-local functional on a Hilbert space of path acts
and identifies its coefficient sequence from a maintained stochastic-choice
link. Neither result substitutes for the other. \citeauthor{LuSaito2026RepeatedChoice} do not impose the
abstract-Wiener path structure or separate signed from admissible weights; the
present paper does not derive an ergodic Markov law. A natural extension, left
for future work, is to impose repeated-choice stationarity conditions on
path-valued acts and ask when the induced ergodic frequencies identify the
admissible normalization of the Wiener weights.

\tab For a normalized state-feature coordinate amplitude $s_{1,n}$, we
interpret $P_n(A)=\int_A|s_{1,n}(x)|^2dx$ as the probability assigned to the
outcome event $A$. The squared-amplitude map is Born-rule-type in form but not
in derivation: it is an interpretive assumption about how state-feature
amplitudes translate into outcome densities, not a theorem
\citep[cf.][]{CharlesCadogan2018}. This outcome-density map is distinct from
the positive-part normalization $p_j=v_j^+/\sum_kv_k^+$ used to transform
signed Wiener weights into admissible decision weights.

\tab The path-space formulation is useful when a finite-state representation
would suppress economically relevant information about timing, persistence,
or ordering. A finite lottery remains appropriate when primitive uncertainty
consists of a small set of consequences with stable additive probabilities.
The richer construction is warranted when the act itself is function-valued:
a Banach path space describes the act, a Hilbert subspace supplies tractable
local coordinates, and the Gaussian reference measure supplies a benchmark
law for local stochastic perturbations. That reference measure is not an
elicited belief; it is a modelling benchmark, analogous to Brownian
perturbations in continuous-time finance, macroeconomics, and stochastic
control. The framework is useful when the analyst wants a local representation
of preferences over stochastic paths while allowing state dependence,
nonlinear weighting, or rank-dependent restrictions. It is not meant to
replace ordinary SEU in small finite-state problems
\citep{Savage1972,AnscombeAumann1963,Kreps1988}.

\tab The generic Hilbert-space formulation connects the paper to the classical
economics of continuous utility and infinite-dimensional commodity and claim
spaces. Continuous numerical representation begins with \citet{Debreu1954};
\citet{Grandmont1972} studies continuity for von Neumann--Morgenstern utility;
\citet{Chichilnisky1980} develops a continuous representation of preferences
using a topology defined directly on the preference space; and
\citet{Bewley1972} shows why topology and preference continuity are
economically consequential when commodities are infinite dimensional.
Continuous-time consumption and asset-pricing models likewise treat streams
or contingent claims as function-valued objects
\citep{DuffieZame1989,HindyHuang1992,BackPliska1991}. The present result does
not reproduce their equilibrium or no-arbitrage theorems. It characterizes a
projected-local preference functional and states the design conditions under
which its coefficient sequence is identified. The Gaussian and Wiener
structure enters only in the subsequent application.

\tab This interpretation clarifies what can be tested. Once a finite basis is
chosen, choices over projected acts identify local utility coefficients and
normalized decision weights up to the usual location and scale
normalizations. Estimated weights must be nonnegative and add to one after
normalization; RDU and CPT impose rank-order restrictions; and the signed
Wiener representation can be rejected if the recovered positive-part weights
fail monotonicity or if additional basis terms do not improve out-of-sample
choice fit. \Cref{thm:FiniteProjectedElicitation} formalizes the finite
projected representation and elicitation argument, while
\Cref{cor:SEURDUCPTRestrictions} gives the corresponding SEU, RDU, and CPT
overidentifying restrictions. Supplementary derivations—including cylindrical-set representations, repetition-level law of the iterated logarithm bounds, conditional Gaussian reweighting diagnostics, bounded-window quadrature tables, and Malliavin-derivative extensions—are gathered in Online Appendices~A.1--A.5.
\section{Hilbert-Space Theory and Wiener Specialization}\label{sec:TheModel}
This section first states the representation and identification theory on a generic separable real Hilbert space. It then specializes that theory to abstract Wiener space and extends the local preference functional of \citet{Machina1982} to path-valued acts. The Hilbert-space results do not require a Gaussian measure; the Wiener application supplies the path space, Gaussian benchmark, Hermite coordinates, and signed stochastic weights used in \cref{sec:Applications}.
Write \(\langle a,b\rangle_{\mathcal H}\) for the inner product on \(\mathcal{H}\), with induced norm \(\|a\|_{\mathcal H}=\sqrt{\langle a,a\rangle_{\mathcal H}}\). Relative to the orthonormal basis \(\{e_j\}\), \(\langle a,b\rangle_{\mathcal H}=\sum_{j\ge1}a_jb_j\).
\begin{defn}[Decision-theoretic primitives]\label{defn:DecisionPrimitives}
	Let $\mathcal H$ be a separable real Hilbert space with orthonormal basis $\{e_j\}_{j\geq1}$. An \emph{act} $f_a$ is represented by its loading vector $a=\sum_{j\geq1}a_je_j\in\mathcal A\subset\mathcal H$. The basis elements are state features or coordinates of the economic act, not objects of preference. A projected-local representation evaluates $f_a$ through a continuous linear functional $V(a)=\langle\theta,a\rangle_{\mathcal H}$. In the abstract Wiener application below, $\mathcal H$ becomes a Cameron--Martin space, $e_j$ becomes the Hermite coordinate $s_{1j}$, and the raw Wiener coefficient $v_j$ is a signed local change-of-measure coefficient. An admissible decision weight is then a normalized nonnegative transform, for example $p_j=v_j^+/\sum_k v_k^+$ on the event $\sum_kv_k^+>0$, with an explicit fallback on its complement in a finite model.\qed
\end{defn}

\tab Preferences are defined over acts, not over basis elements or states themselves. Thus a statement such as ``the decision maker prefers feature $e_j$ to $e_k$'' is shorthand for a comparison of acts whose payoff variation loads more heavily on $e_j$ than on $e_k$. This convention preserves the distinction between states and objects of choice while allowing basis features to affect local valuation. The stochastic interpretation of those coordinates is additional structure supplied by the Wiener application.
\subsection{Generic Hilbert-space axiomatization and identification}
\label{subsec:InfiniteAxiomatization}

\tab We now state a complete characterization for the continuous linear
projected-local class. Let \(\mathcal{H}\) be the real separable Hilbert space
with orthonormal basis \(\{e_j\}_{j\geq1}\), let
\(\mathcal A\subset\mathcal H\) be convex with nonempty relative interior and
contain the zero act, and write \(f_{a}\) for the act with loading
\(a=(a_1,a_2,\ldots)\in\mathcal A\). Let \(P_{J}\) denote orthogonal projection
onto the first \(J\) coordinates. Let \(K\subset\mathcal H\) be a closed convex
cone of loading changes interpreted as monetary improvements.
\begin{axiom}[Hilbert weak order]\label{ax:HilbertWeakOrder}
	The relation \(\succeq \) on \(\mathcal{A}\) is complete, transitive, and
	nontrivial.
\end{axiom}
\begin{axiom}[Mixture independence]\label{ax:HilbertIndependence}
	For all \(a,b,c\in\mathcal A\) and \(\alpha\in(0,1)\), whenever the mixtures
	belong to \(\mathcal{A}\),
	\(f_{a}\succeq f_{b}\quad \iff \quad f_{\alpha a+(1-\alpha )c}\succeq f_{\alpha b+(1-\alpha )c}.\)
\end{axiom}
\begin{axiom}[Hilbert continuity]\label{ax:HilbertContinuity}
	For every \(a\in\mathcal A\), the upper and lower contour sets are closed in
	the \(\mathcal{H}\)-norm topology.
\end{axiom}
\begin{axiom}[Projective consistency]\label{ax:ProjectiveConsistency}
	For every \(J<M\), the ranking of acts in \(P_J\mathcal A\) is unchanged when
	those acts are regarded as elements of \(P_M\mathcal A\). Moreover, if
	\(P_Ja,P_Jb\in\mathcal A\) for all sufficiently large \(J\), then
	\(P_{J}a\succeq P_{J}b\ \text{eventually\ and\ }P_{J}a\rightarrow a,\ P_{J}b\rightarrow b\quad \implies \quad f_{a}\succeq f_{b}.\)
	The same implication holds with strict preference whenever the limiting value
	difference is bounded away from zero.
\end{axiom}
\begin{axiom}[Cone monotonicity]\label{ax:ConeMonotonicity}
	If \(a-b\in K\), then \(f_a\succeq f_b\); if \(a-b\) belongs to the relative
	interior of \(K\), then \(f_a\succ f_b\).
\end{axiom}
\begin{thm}[Hilbert-space projected-local representation]
	\label{thm:InfiniteProjectedRepresentation}
	Axioms~\ref{ax:HilbertWeakOrder}--\ref{ax:ConeMonotonicity} hold if and only if there is a
	nonzero coefficient sequence \(\theta\in K^*\subset\mathcal H\), where
	\(K^*=\{h:\langle h,k\rangle_{\mathcal H}\geq0\text{ for all }k\in K\}\),
	with \(\langle\theta,k\rangle_{\mathcal H}>0\) for every
	\(k\in\operatorname{ri}(K)\setminus\{0\}\), where
	\(\operatorname{ri}(K)\) denotes the relative interior of the improvement cone,
	such that
	\(f_{a}\succeq f_{b}\quad \iff \quad V(a)=\langle \theta ,a\rangle _{\mathcal{H}}\ge \langle \theta ,b\rangle _{\mathcal{H}}=V(b).\)
	Equivalently,
	\(V(a)=\sum _{j=1}^{\infty }\theta _{j}a_{j},\qquad \sum _{j=1}^{\infty }\theta _{j}^{2}<\infty ,\)
	with convergence uniform on norm-bounded sets after projection error is
	controlled by \(\|a-P_Ja\|_{\mathcal H}\). The representing coefficient is
	unique up to a positive scalar. Under \(\|\theta\|_{\mathcal H}=1\) and the
	orientation imposed by \(K\), it is unique. \qed
\end{thm}
\begin{proof}
	Necessity follows immediately from a nonzero \(\theta\in K^*\): the inner
	product defines a nontrivial continuous affine functional, its restrictions to
	the nested projection spaces are consistent, and membership of \(\theta \) in
	the dual cone gives monotonicity.
	For sufficiency, mixture independence, the weak-order axioms, and continuity
	give an affine representation on the convex act domain. Normalize its value at
	the zero act to zero. Its linear part is continuous in the Hilbert norm;
	therefore, by the Riesz representation theorem, there is a unique
	\(\theta\in\mathcal H\) such that \(V(a)=\langle\theta,a\rangle_{\mathcal H}\).
	Nontriviality implies \(\theta\neq0\). Projective consistency makes the
	finite-dimensional representing vectors the coordinates
	\((\theta_1,\ldots,\theta_J)\) of this same element, rather than unrelated
	representations, and Hilbert continuity extends their values from the dense
	union of projection spaces to \(\mathcal{A}\). Cone monotonicity implies
	\(\langle\theta,k\rangle\geq0\) for every \(k\in K\), hence
	\(\theta\in K^*\), while strict monotonicity gives
	\(\langle\theta,k\rangle>0\) on
	\(\operatorname{ri}(K)\setminus\{0\}\). Two continuous linear functionals representing the same
	nontrivial mixture order are positive scalar multiples; unit-norm and cone
	orientation remove that scalar ambiguity. Finally, Cauchy--Schwarz gives
	\(|V(a)-V(P_{J}a)|\le \min \!\left\{\|\theta \|_{\mathcal{H}}\|a-P_{J}a\|_{\mathcal{H}},\|\theta -P_{J}\theta \|_{\mathcal{H}}\|a\|_{\mathcal{H}}\right\},\)
	which proves pointwise convergence and uniform convergence on norm-bounded
	subsets of \(\mathcal{A}\).
\end{proof}
\paragraph{The design operator and variation in observed comparisons.}
The design is the distribution of pairs of acts $(f_a,f_b)$ presented
for comparison. Write $D=a-b$ for the resulting random difference in
their loading vectors. For any direction $h\in\mathcal H$, define the
design operator $Q:\mathcal H\to\mathcal H$ by
\(Qh=\mathbb{E}\!\left[\langle h,D\rangle _{\mathcal{H}}D\right],\)
where expectation is taken over the distribution of presented comparisons.
A sufficient condition for this Hilbert-valued expectation to exist and
define a bounded linear operator is
$\mathbb E|D|{\mathcal H}^{2}<\infty$, since
$\mathbb E|\langle h,D\rangle{\mathcal H}D|{\mathcal H}
\leq|h|{\mathcal H}\mathbb E|D|_{\mathcal H}^{2}$.
In finitely many orthonormal coordinates, $Q$ is the second-moment
matrix $\mathbb E[DD^{\top}]$; it equals the covariance matrix only
when $\mathbb E[D]=0$.
The identity
\(\langle h,Qh\rangle _{\mathcal{H}}=\mathbb{E}\!\left[\langle h,D\rangle _{\mathcal{H}}^{2}\right]\)
measures the variation in comparisons along direction $h$. If it is
zero, then $\langle h,D\rangle_{\mathcal H}=0$ almost surely, so
changing a preference coefficient in direction $h$ leaves all observed
choice indices unchanged. Injectivity on the maintained parameter space
means that no nonzero direction in that space is invisible to the
design. For example, if every comparison varies only the first loading,
choices cannot identify a coefficient on the second loading without
further restrictions. The theorem below therefore restricts the
parameter to the closed span of the observed differences and states
the corresponding injectivity condition explicitly.
\begin{thm}[Identification from stochastic repeated choices]
	\label{thm:InfiniteIdentification}
	Suppose binary menu observations satisfy
	\(\Pr (f_{a}\text{\ chosen\ over\ }f_{b}\mid a,b)=G\!\left(\langle \beta ,a-b\rangle _{\mathcal{H}}\right),\qquad \beta =\gamma \theta ,\quad \gamma >0,\)
	where \(G\) is known, continuous, and strictly increasing. Let \(D=a-b\), and
	define the maintained parameter space
	\(\mathcal{H}_{D}=\overline{\mathrm{span}\>}\bigl(\>\mathrm{supp}\>D\bigr)\subseteq \mathcal{H},\)
	and restrict \(\beta\in\mathcal H_D\). Suppose the design operator
	\(Qh=\mathbb{E}[\langle h,D\rangle _{\mathcal{H}}D]\)
	is well defined and injective on \(\mathcal{H}_{D}\). Then
	\(\beta \) is point identified from the population choice probabilities. If
	\(\|\theta\|_{\mathcal H}=1\), then
	\(\gamma=\|\beta\|_{\mathcal H}\) and
	\(\theta=\beta/\|\beta\|_{\mathcal H}\) are separately identified.
	If \(\theta_j=p_ju_j\), the normalized weights \(p_{j}\) are separately
	identified when the local utility sequence \(u_{j}\) is externally elicited,
	\(u_j\neq0\) on active coordinates, and
	\(0<\sum_j\theta_j/u_j<\infty\); in that case
	\(p_{j}=\frac{\theta _{j}/u_{j}}{\sum _{k=1}^{\infty }\theta _{k}/u_{k}}.\)
	Without an external normalization of \(u\), choice data identify only the
	composite sequence \(\theta \): for any positive sequence \(c_{j}\) preserving
	summability, \((p_jc_j,u_j/c_j)\) generates the same composite coefficients
	after renormalization. \qed
\end{thm}
\begin{proof}
	Strict monotonicity of \(G\) identifies
	\(\langle\beta,D\rangle_{\mathcal H}\) for every design difference in the
	support. If \(\beta \) and \(\widetilde{\beta }\) generate the same probabilities,
	then \(\langle \beta -\widetilde{\beta },D\rangle =0\) almost surely. Consequently,
	\(\langle \beta -\widetilde{\beta },Q(\beta -\widetilde{\beta })\rangle =\mathbb{E}[\langle \beta -\widetilde{\beta },D\rangle ^{2}]=0.\)
	Injectivity of \(Q\) implies \(\beta =\widetilde{\beta }\). Unit-norm
	normalization then separates scale from direction. The displayed ratio follows
	from \(\theta_j=p_ju_j\) and \(\sum_jp_j=1\). The final transformation proves
	the nonidentification statement.
\end{proof}
\begin{rem}[Relation to \citeauthor{LuSaito2026RepeatedChoice}]\label{rem:LuSaitoComparison}
	\citet{LuSaito2026RepeatedChoice} derive ergodic choice frequencies from a
	dynamic Markov utility process and characterize that process by four behavioral
	axioms. \Cref{thm:InfiniteProjectedRepresentation} instead characterizes the
	continuous linear local-value functional on a Hilbert space of acts, while
	\Cref{thm:InfiniteIdentification} identifies its coefficient sequence from a
	maintained stochastic-choice link. Neither contribution substitutes for the
	other: Lu and Saito do not impose the abstract-Wiener path structure or
	distinguish signed from admissible weights, whereas the present paper does not
	derive an ergodic Markov law.
\end{rem}
\subsection{Abstract Wiener-space application}\label{subsec:ExtendMachina1982}

\tab The preceding representation and identification theorems are Hilbert-space results: neither invokes a Gaussian measure, a Cameron--Martin embedding, nor Hermite coordinates. The present subsection adds those structures. Their role is load-bearing for the stochastic path interpretation, the Wiener-integral coefficient construction, and the distinction between signed analytical weights and admissible decision weights, but not for the Riesz representation itself.
\paragraph{Canonical outcome domain and coordinate conventions.}
For an outcome $x\in[0,L]$, with upper endpoint $L>0$, define
\begin{equation}\label{eq:CanonicalRescaling}
 t=\frac{x}{L},\qquad F_L(t):=F(Lt),\qquad U_L(t):=U(Lt),\quad t\in[0,1].
\end{equation}
Thus $F_L$ is the distribution function of the normalized outcome,
and $U_L$ preserves utility at corresponding outcomes. Separately,
for $g\in L^2[0,L]$, define $g_L(t)=\sqrt L\,g(Lt)$.
The map $R_L:g\mapsto g_L$ preserves the $L^2$ norm because
$\int_0^1|g_L(t)|^2dt=\int_0^L|g(x)|^2dx$.
The factor $\sqrt L$ applies to this norm-preserving function map,
not to $F_L$. Take $L=M$ for $D[0,M]$; $L=1$ and $L=2$ denote
distinct truncation experiments. The auxiliary real coordinate $x$
is used for full-line Hermite identities, while $\omega$ denotes a random path.
\citet{Machina1982} studies a possibly nonlinear preference functional
$V$ on the distribution-function set $D[0,1]$. Here distributions are
identified with their $L^2[0,1]$ equivalence classes; right-continuity and
the distribution boundary conventions determine their usual versions.
They lie in the separable real Banach space
$B=L^2[0,1]$. The notation $L_B[0,1]$ in the Wiener specialization will
mean this declared Banach ambient space, not all bounded measurable
functions with the sup norm, which is generally nonseparable.
\paragraph{A concrete Wiener embedding.}
Let $\{\phi_j\}_{j\geq1}$ be an orthonormal basis of $B$, and choose
$\gamma_j>0$ with $\sum_j\gamma_j<\infty$. Define
\(H=\widetilde{L}_{H}^{2}[0,1]=\left\{h=\sum _{j}h_{j}\phi _{j}:\sum _{j}h_{j}^{2}/\gamma _{j}<\infty \right\},\qquad \|h\|_{H}^{2}=\sum _{j}h_{j}^{2}/\gamma _{j}.\)
The inclusion $\iota:H\hookrightarrow B$ is continuous and dense;
$e_j=\sqrt{\gamma_j}\phi_j$ is orthonormal in $H$. For independent
standard normal variables $Z_j$, the series
$W=\sum_j\sqrt{\gamma_j}Z_j\phi_j$ converges in $L^2(\Omega;B)$ and
almost surely. Its law $P_\star$ defines an abstract Wiener triple
$(\iota,H,B)$, with $E\|W\|_B^2=\sum_j\gamma_j$.
\onehalfspacing\begin{defn}[Abstract Wiener space]\label{defn:AbstractWienerSpace}\cite{CameronMartin1944,Gross1967}.\\
	Let $\Omega$ be a separable Banach space, $P_\star$ be a probability measure on $\Omega$, and $\mathcal{F}$ be the $\sigma$-field of Borel measurable subsets of $\Omega$. There exist a separable Hilbert space $\mathfrak{H}$ that is continuously and densely embedded in $\Omega$ with inclusion map $\iota:\mathfrak{H}\rightarrow \Omega$ and such that
	\begin{equation}
		\displaystyle\int_\Omega e^{i y(x)}P_\star(dx)=e^{-\frac{1}{2}\|\iota^*y\|^2_{\mathfrak H}}
	\end{equation}
	for any $y\in\Omega^\star$, with the norm in the exponent understood as
	$\|\iota^*y\|_{\mathfrak H}$, where $\iota^*y$ is the Riesz vector of
	$h\mapsto y(\iota h)$. Here $\Omega^\star$ and $\mathfrak H^\star$
	are the respective dual spaces. The triple $(\Omega,\mathfrak{H},P_\star)$
	is called an abstract Wiener space.\qed
\end{defn}
\begin{thm}[Abstract Wiener space over Banach space]\citep[Thm~4.4,~p. 79]{Kuo1975}\label{thm:AbstractWinerSpace}
	Let $\mathfrak{B}$ be a real separable Banach space. Then there exists a separable Hilbert space $\mathfrak{H}$ densely embedded in $\mathfrak{B}$ such that the $\mathfrak{B}$-norm is measurable over $\mathfrak{H}$. Equivalently, the inclusion map $\iota:\mathfrak{H}\to\mathfrak{B}$ supports a Gaussian measure $P_\star$ on $\mathfrak{B}$, and $(\mathfrak{B},\mathfrak{H},P_\star)$ is an abstract Wiener space. \qed
\end{thm}\onehalfspacing
\begin{cor}[Abstract Wiener specialization]\label{cor:AbstractWienerSpecialization}
	Let $(\mathfrak B,\mathfrak H,P_\star)$ be an abstract Wiener space, let $\mathcal A\subset\mathfrak H$ satisfy the domain conditions in \Cref{thm:InfiniteProjectedRepresentation}, and let $\{s_{1j}\}_{j\geq1}$ be an orthonormal coordinate system for $\mathfrak H$. Then the representation and identification conclusions of \Cref{thm:InfiniteProjectedRepresentation,thm:InfiniteIdentification} apply with $\mathcal H=\mathfrak H$ and $e_j=s_{1j}$. The Gaussian measure $P_\star$ additionally permits construction of stochastic coordinate perturbations as Wiener integrals; these are distinct from the fixed representing coefficients. Those integrals may be signed; they are admissible decision weights only after the nonnegative normalization stated below. \qed
\end{cor}

\tab This corollary isolates the division of labor. Hilbert geometry supplies existence, uniqueness, projection, and identification. The abstract Wiener triple supplies stochastic path semantics and a canonical Gaussian benchmark from which Wiener-integral weights can be constructed.
\begin{thm}[Conditional extension of the preference functional]
	\label{thm:MachinaFuncAbstWienerSpace}
	Let $D=D[0,1]\subset B=L^2[0,1]$ and let $(\iota,H,B,P_\star)$ be
	the Wiener construction above. Suppose $V:D\to\mathbb R$ is
	$C$-Lipschitz in the $B$ norm for some finite $C>0$; that is,
	$|V(F)-V(G)|\leq C\|F-G\|_B$ for $F,G\in D$. Then
	\begin{equation}\label{eq:ExplicitPreferenceExtension}
		\overline V(b)=\inf_{F\in D}\{V(F)+C\|b-F\|_B\},\qquad b\in B,
	\end{equation}
	is a finite, $C$-Lipschitz extension of $V$ by the standard infimum extension construction, and
	$\overline V(W)\in L^2(P_\star)$. Hence the original rankings on $D$
	are preserved. No uniqueness, monotonicity outside $D$, or mixture
	independence of this nonlinear extension is asserted.
	If, additionally, $\overline V$ is Fr\'{e}chet differentiable at a
	reference $F_0\in D$, its derivative $\ell_{F_0}\in B^*$ restricts to
	a continuous linear functional on $H$, with Riesz vector
	$\theta_{F_0}\in H$. The derivative, rather than the entire nonlinear
	functional, has the local Hilbert representation.
\end{thm}
\begin{proof}
	Fix $F_0\in D$. Lipschitz continuity and the triangle inequality imply
	\(V(F_{0})-C\|b-F_{0}\|_{B}\le V(F)+C\|b-F\|_{B}\quad (F\in D).\)
	Taking the infimum and using $F=F_0$ proves finiteness and
	$|\overline V(b)-V(F_0)|\leq C\|b-F_0\|_B$.
	For $G\in D$, the Lipschitz inequality gives
	$V(F)+C\|G-F\|_B\geq V(G)$, while $F=G$ gives equality in the
	infimum. Thus $\overline V(G)=V(G)$.
	For $b,c\in B$, each term in the infimum at $b$ is at most its
	counterpart at $c$ plus $C\|b-c\|_B$; taking infima and interchanging
	$b,c$ proves $C$-Lipschitz continuity and hence Borel measurability.
	The preceding bound and
	$E\|W-F_0\|_B^2<\infty$ give square integrability.
	Finally, continuity of $\iota$ gives
	$|\ell_{F_0}(\iota h)|\leq\|\ell_{F_0}\|_{B^*}\|\iota\|\|h\|_H$.
	Riesz representation therefore supplies $\theta{F_0}\in H$ with
	$\ell{F_0}(\iota h)=\langle\theta_{F_0},h\rangle_H$.
	The differentiability remainder is $o(|k|_B)$ at $F_0$; along
	$k=\iota h$ it is also $o(|h|_H)$. This proves the claimed local
	representation.
\end{proof}
\noindent This is a conditional extension result, not a consequence of
the existence of a Gaussian measure alone. Fr\'{e}chet differentiability
on the original distribution domain does not by itself imply the global
$L^2$-Lipschitz hypothesis, nor does Hahn--Banach extend an arbitrary
nonlinear $V$.

\tab Under the one-dimensional Gaussian projection used below, the Hermite functions generated by the Hermite polynomials ${H_n(x)}{n=0}^{\infty}$ form an orthonormal basis for $L^2(\mathbb{R})$ \citep[see~e.g.][pg.~25]{AkheiGlaz1961} where
\begin{equation}
	H_n(x) = (-1)^n;e^{x^2}\frac{d^n;e^{-x^2}}{dx^n}
\end{equation}
\begin{figure}[!htb!]
	\centering
	\begin{minipage}[h]{0.4\linewidth}
		\captionof{figure}{Gauss--Hermite state-feature densities}
		\label{fig:AbsWienerGaussHkdens}
		\centerline{\includegraphics[scale=.8]{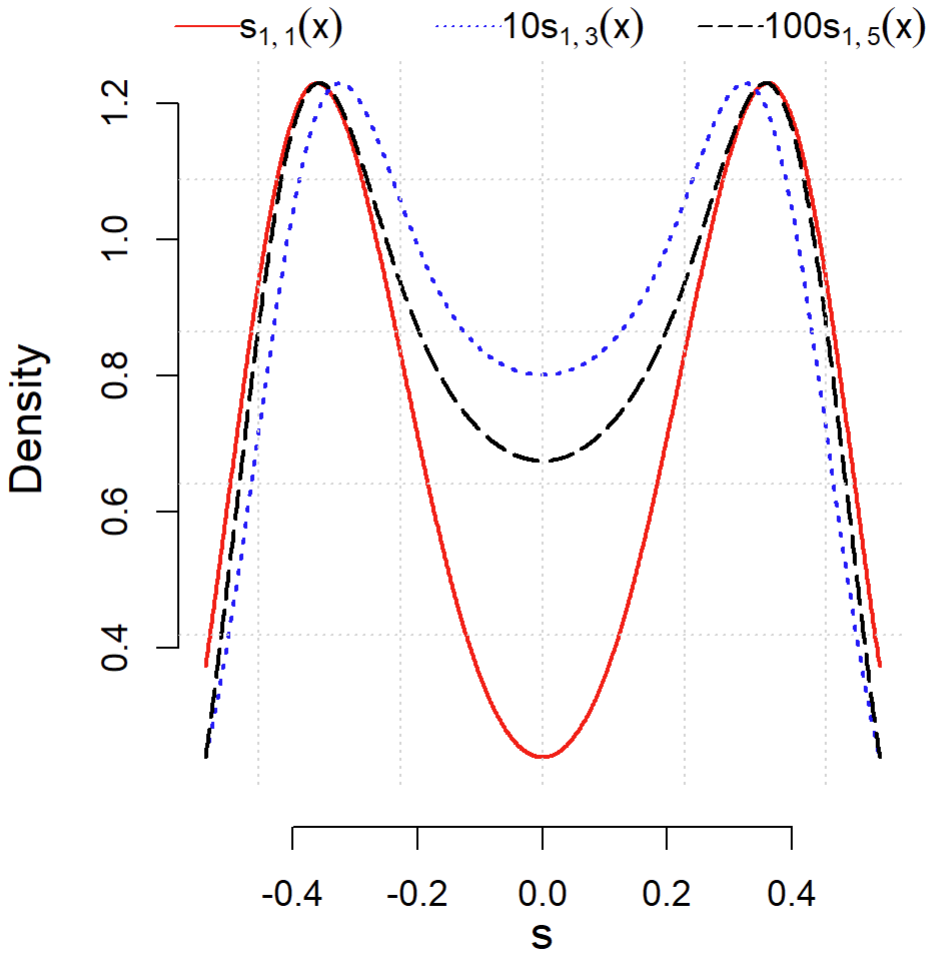}}
	\end{minipage}
	\vspace*{.5cm}
	\begin{minipage}[h]{0.9\textwidth}
		\footnotesize\textit{Notes:} \cref{fig:AbsWienerGaussHkdens} shows the normalized Gauss--Hermite density obtained for 200 evenly spaced points in the interval $[-2,2]$. The plots are for $s_{11}, 10s_{13}, 100s_{15}$ where scales 10 and 100 are used to harmonize the $s$ axis range. The horizontal coordinate here is the auxiliary full-line coordinate $x$, not the canonical outcome $t$. For a bounded window $[0,L]$, the induced density is $L|s_{1,n}(Lt)|^2/\int_0^L|s_{1,n}(x)|^2dx$ on $t\in[0,1]$. Display multipliers 10 and 100 do not define normalized densities. \end{minipage} \end{figure} \noindent Full-line Hermite features use the auxiliary real coordinate $x$. They are mapped to a bounded outcome window by \eqref{eq:CanonicalRescaling}; restriction does not preserve orthogonality. The normalized full-line features are \begin{equation} s_{1,n}(x) = A_nH_n(x)e^{-\frac{x^2}{2}}\label{eq:BehavioralBasisFunction} \end{equation} which upon normalization (see e.g. \citet[p.~57]{Wiener1933}, \citet[Appendix~A]{boyd2018}) gives \begin{equation*} \int^{\infty}_{-\infty}|s_{1,n}(x)|^2dx = A^2_n\displaystyle\int^{\infty}_{-\infty}H^2_n(x)e^{-x^2}dx=1,\qquad A_n = \bigl(2^n n!\sqrt{\pi}\bigr)^{-1/2} \end{equation*} Here $|s_{1,n}(x)|^2$ is a state dependent Gauss--Hermite probability density function on $\mathbb{R}$.\footnote{See \citet[p.~181]{Kreyszig1978} where Gram-Schmidt orthogonalization was applied to the product of a Gaussian density $e^{-x^2/2}$ and Hermite polynomial $H_n(x)$ to derive an orthogonal basis like $\{s_{1j}\}_{j=1}^\infty$.} A plot of the densities for those functions is shown in \cref{fig:AbsWienerGaussHkdens}. These functions are basis features of payoff states; they are not chosen by the decision maker. A risk-averse decision maker may evaluate acts differently when their payoff variation loads on $s_{11}$ rather than $s_{13}$ or $s_{15}$, but the preference relation remains a relation over acts. On the canonical bounded domain use $\psi_{j,L}(t)=\sqrt{L}\,s_{1,j}(Lt)$. A declared Hilbert isometry transfers an orthonormal $L^2(\mathbb R)$ system to $H$; raw bounded restrictions are not assumed orthonormal in $H$.
\subsubsection{Vector valued utility as abstract Fourier coefficients in Hilbert space}\label{subsubsec:VectorValuedUtility}
Let $U(x)$ be a utility function defined on $X$ taking values in $\mathbb{R}$, and $x\in X$ i.e. $U:X\to \mathbb{R}$. The abstract Fourier coefficient \citep[p.~89]{Haase2010} is given by
\begin{align}
	u_j(x,s) &= \frac{<U,;s_{1,j}>}{<s_{1,j},;s_{1,j}>}=\frac{\displaystyle \int_{z\in X} U(z),s_{1,j}(z)dz}{\displaystyle\int_{z\in X} |s_{1,j}(z)|^2;dz}\label{eq:AbstractFourierCoeff}
\end{align}
For an orthogonal family the displayed quotient is its Fourier
coefficient. With several bounded restricted features
$\psi_{j,L}$ one must solve the Gram system
\begin{equation}\label{eq:BoundedGramProjection}
	G_Lc_L=b_L,\qquad
	(G_L)_{jk}=\int_0^1\psi_{j,L}(t)\psi_{k,L}(t)dt,\qquad
	(b_L)_j=\int_0^1\sqrt{L}\,U_L(t)\psi_{j,L}(t)dt,
\end{equation}
where $c_L$ contains simultaneous least-squares coefficients and $q_j(L)=(b_L)_j/(G_L)_{jj}$ denotes the separate one-feature projection coefficient.

\tab We use the one-dimensional Gaussian projection of abstract Wiener space to fix the reference measure $P_\star$ and to avoid complex-valued basis functions. Interpreting the normalized squared amplitude $\text{P}_n(A) = \int_{x\in A}|s_{1,n}(x)|^2dx$ as an induced outcome density connects to nonlinear source-dependent probability weighting in prospect theory \citep{AbdellaouiBaillonPlacidoWakker2011} and earlier Born-rule-type constructions \citep{WangBusemeyerAtamanspacherPothos2013,CharlesCadogan2018}. Expanding $U(x;\mathbf{u}) = \sum_{j=0}^{\infty}u_js_{1,j}(x)$ and normalizing $\widehat{u}_j=u_j/\|\mathbf{u}\|$ induces a probability measure on utility outcomes that is equivalent to a cylindrical-set distribution on the coordinate sequence $(u_0,u_1,\ldots)$; see Appendix~\ref{app:CylindricalBornRule} for the Riesz--Fischer normalization and odd-feature restriction details \citep[cf.][]{FriedmanSavage1948,Markowitz1952}.
\paragraph{Deterministic values and Gaussian perturbations are different objects.}
The deterministic local functional is scalar-valued; it is not an element
of $H$. Its Riesz vector is. An isonormal Gaussian process on $H$ is a
linear isometry $I:H\to L^2(P_\star)$ satisfying
$E[I(h)]=0$ and $E[I(h)I(k)]=\langle h,k\rangle_H$.
It provides a stochastic perturbation, not an automatic representation
of the same deterministic preference value.
\begin{thm}[Wiener-coordinate expansion of a local perturbation]
	\label{thm:MachinaInfDimPreFunc}
	Let $\{e_j\}$ be orthonormal in $H$, let
	$h=\sum_jh_je_j\in H$ be deterministic, and put
	$v_j(\omega)=I(e_j)(\omega)$. Then
\[
	Z_h(\omega)=I(h)(\omega)=\sum_{j\geq1}h_jv_j(\omega)
	\quad\hbox{in }L^2(P_\star),\qquad
	E[Z_h]=0,\quad E[Z_h^2]=\|h\|_H^2.
\]
	The variables $v_j$ are independent standard normal coordinates. For a
	nonorthonormal feature family $g_j\in H$, their covariance is instead
	$E[I(g_j)I(g_k)]=\langle g_j,g_k\rangle_H$; convergence requires
	convergence of the corresponding vector series in $H$.
\end{thm}
\begin{proof}
	For $h^{(J)}=\sum_{j=1}^Jh_je_j$, linearity gives
	$I(h^{(J)})=\sum_{j=1}^Jh_jI(e_j)$.
	The isometry gives
	$E|I(h)-I(h^{(J)})|^2=\|h-h^{(J)}\|_H^2\to0$.
	The covariance identity gives the stated variance and, since finite
	coordinate vectors are Gaussian, orthogonality gives independence.
	For general features the same isometry gives their Gram covariance.
\end{proof}
\noindent Deterministic coefficient sequences are maintained in this
theorem. Random coefficients need additional integrability and a specified
stochastic-integration convention; they cannot simply be inserted into
the isometry. In particular, $E[v_j]=0$, so raw Wiener weights alone
cannot reproduce a nontrivial deterministic ex ante utility without an added mean or normalization.
\begin{cor}[Wiener integral]\label{cor:WienerIntegral}
	The local perturbation $Z_h$ is a Wiener integral. Its signed coordinates
	are not probabilities. For a finite projection, positive-part weights
	$p_j=v_j^+/\sum_{k=1}^Jv_k^+$ are defined on the event that the
	denominator is positive. A complete random-preference model must also
	specify weights on its complement, for example $p_j=1/J$ there.
\end{cor}
\begin{cor}[Isometric and Gram bounds]\label{cor:EstimatedVF}
	For deterministic orthonormal coefficients,
	$\|\sum_{j=1}^Jh_jv_j\|_{L^2(P_\star)}^2=\sum_{j=1}^Jh_j^2$.
	For features with Gram matrix $G$ and coefficient vector $c$,
	\(\left\|\sum _{j=1}^{J}c_{j}I(g_{j})\right\|_{L^{2}(P_{\star })}^{2}=c^{\top }Gc\le \lambda _{\max }(G)\|c\|_{2}^{2},\)
	where $\lambda_{\max}(G)$ is the largest eigenvalue.
\end{cor}
\begin{proof}
	Expand the square and use the covariance identity in
	\Cref{thm:MachinaInfDimPreFunc}; the final inequality is the Rayleigh
	quotient bound. No identity of the form $|dP_\star|^2=dx$ is used.
\end{proof}
\paragraph{Ex ante order and realized order.}
For a finite projection, take deterministic utility scales $u_j$ and
one common random weight vector $p(\omega)$ for the entire menu. Define
\begin{equation}\label{eq:ExAnteRealizedOrders}
	\begin{aligned}
		V_\omega(a)&=\sum_{j=1}^Jp_j(\omega)u_ja_j,\\
		V^{\mathrm{ex\ ante}}(a)&=E[V_\omega(a)]
		=\sum_{j=1}^J\theta_ja_j,\qquad \theta_j=E[p_j]u_j.
	\end{aligned}
\end{equation}
Both $a\succeq_\omega b$ iff $V_\omega(a)\geq V_\omega(b)$ and
$a\succeq b$ iff $V^{\mathrm{ex\ ante}}(a)\geq V^{\mathrm{ex\ ante}}(b)$ are transitive
weak orders. Nontriviality and the maintained cone-monotonicity
conditions must still be checked. In an infinite projection also require
$\theta\in H$ and sufficient summability to interchange expectation
and the series. Random reversals across different draws of $\omega$
do not violate transitivity within any fixed draw. The Gaussian
construction is an optional random-preference specification; it does
not force the coefficients of the deterministic representation theorem
to be random.
\section{Applications and Numerical Experiment}\label{sec:Applications}\onehalfspacing
\subsection{Subjective expected utility}\label{subsec:SEU}
This example is motivated by \citet{Karni2017} who argued that state dependent utility functions are supported by \mycite{Savage1972} postulates. In the context of \cite{Savage1972} Subjective Expected Utility Theory (SEU) there exist a set of acts $F$ (by abuse of notation this is different from the distribution function $F$), a (discrete) set of states $\emph{S}$, and a subjective probability distribution $\pi$ on $\emph{S}$ such that $f\in F$, $s\in \emph{S}$, and $f(s)=x$. Acts map states into a space of outcomes $X$ called consequences, so $F=X^S$, and decision makers evaluate $\sum_s \pi(s)u(f(s))$. In SEU, $u(f(s))$ represents state-independent ``taste'' over consequences, while $\pi(s)$ captures beliefs about which state occurs. In our setup in \Cref{thm:MachinaInfDimPreFunc}, $u(f(s))$ is replaced by state-dependent local utility $u_j(s)$, and in the optional random-weight specification, the scalar belief is replaced by the explicitly normalized finite vector $p(\omega)$, whereas the raw Gaussian coordinates $v_j(\omega)$ describe a signed perturbation.

\tab In SEU theory, beliefs are probabilities on states separated from tastes over consequences \citep[see][and references therein]{Baccelli2020}, while our $v_j(\omega)$ is state dependent by construction. Analogous to the kernel reweighting technique in \citet[p.~47]{RaiffaSchlaifer1961}—where a nonnegative kernel $\kappa(s\mid P^\star)$ normalized via $N(s)$ produces a posterior density $f(s\mid P^\star)=N(s)\kappa(s\mid P^\star)$ with $\int f(s\mid P^\star)\,dP^\star(s)=1$—the normalized positive part of the Wiener weight acts as a signed change of measure normalized into a choice probability measure. Because these normalized weights depend on the state functions, they generally violate the separability between beliefs and tastes required by \citet{Savage1972}, \citet[pp.~34--35]{Kreps1988}, and \citet{AbdellaouiWakker2020}; thus, Savage's SEU applies here only after replacing the state-dependent weights with an exogenous probability measure such as $P_\star$.
\subsection{Cumulative Prospect Theory}\label{subsec:CPT}
One of the linchpins of \citet{TverKahn1992} Cumulative Prospect Theory (CPT) is the use of \mycite{Quiggin1982} Rank Dependent Utility (RDU) transformation of nonlinear probability into linear decision weights.
The additive linearized weights produced by RDU \citep[Chapters~7--8]{Wakker2010} require ordered, nonnegative probabilities that sum to one. Let $(w\circ v_j)(P_\star)$ be a composite weighting functional for some probability weighting function $w(\mydot)$. Since $v_j$ is a Wiener integral, the composite function is also a Wiener functional; however, $v_j$ need not be nonnegative and need not sum to one across states. Therefore RDU cannot be applied directly to the raw Wiener-integral weights. One must first pass to normalized nonnegative weights, such as $p_j=v_j^+/\sum_kv_k^+$ on the event $\sum_kv_k^+>0$. Once this normalization is imposed, the usual RDU linearization scheme is
\begin{equation*}
	\begin{split}
		&\pi_1 = w(p_1),\quad\pi_2 = w(p_1+p_2)-w(p_1),\quad\pi_3=w(p_1+p_2+p_3)-w(p_1+p_2),\ldots\\
		&\pi_j = w(p_1+p_2+\cdots+p_j)-w(p_1+p_2+\cdots+p_{j-1}),\ldots,\quad\pi_n = 1-w(p_1+p_2+\cdots+p_{n-1})\nonumber
	\end{split}
\end{equation*}
where $w(0)=0$ and $w\left(\sum_{j=1}^{n}p_j\right)=1$. With normalized probabilities the $\pi_j$ sum to one by construction; without that normalization, the raw Wiener weights are not a probability vector and the RDU interpretation fails.

\tab Positive-part normalized weights within a menu sum to one and are therefore dependent; they do not become independent Wiener increments after nonlinear probability weighting. For a smooth weighting function $w$, the local increment $\pi_j = w(Y_{j-1}+p_j) - w(Y_{j-1}) = w'(Y_{j-1})p_j + o(|p_j|)$, where $Y_{j-1} = \sum_{k<j} p_k$, is an ordinary Taylor expansion. When decision weights are sampled independently across repeated choice occasions with finite variance, individual weight components satisfy the classical law of the iterated logarithm \citep[\S10.2]{ChowTiecher1988,Kuelbs1977} with fluctuation scale $\sqrt{2\sigma^2 \log\log n / n}$ (see Appendix~\ref{app:IndependentWeightLIL} for the formal proposition, proof, and centering remarks).
\subsection{Local expected utility and bounded-window projection}\label{subsec:LocalUtilMax}

\tab \citet[pp.~294--295]{Machina1982} proffered $V(\cdot)$ as a nonlinear functional, and posited
\begin{equation}
	V(F^\star)-V(F)= \int U(x;,F)(dF^\star(x)-dF(x))+o\left(|F^\star-F|\right) \label{eq:MachinaLocalUtilDiff}
\end{equation}
as a first order expansion of a Fr\'{e}chet derivative. If $F^\star\succeq F$, then $V(F^\star)-V(F)\geq 0$. In our case, \Cref{thm:MachinaInfDimPreFunc} tells us that
\begin{equation}
	Z_{h^\star}(\omega)-Z_h(\omega)=\sum^{\infty}{j=1},(h^\star_j-h_j)v_j(\omega)\label{eq:AbstractWienerUtilDiff}
\end{equation}
is to be read as a difference of realized local perturbations,
not as an equality between a deterministic Machina value and a random
variable. More explicitly, write
$Z{h^\star}(\omega)-Z_h(\omega)
=\sum_j(h_j^\star-h_j)v_j(\omega)$ when the deterministic coefficient
difference belongs to $H$. For economically normalized realized values,
use $V_\omega(a^\star)-V_\omega(a)$ from
\eqref{eq:ExAnteRealizedOrders}; for ex ante comparisons use its
expectation. Every fixed-$\omega$ comparison uses the same utility
function for all acts and is transitive. Pairwise majority probabilities
$\Pr{V_\omega(a)>V_\omega(b)}$ can have aggregation cycles under
additional distributions of preferences, but no such cycle is proved
here. Neither weak stochastic transitivity nor its failure follows from
the Wiener construction alone. Below, $v_j$ suppresses its $\omega$
argument only as a notational convention.

\tab Machina's local derivative, the deterministic ex ante score, and the random perturbation serve distinct roles: a Wiener integral is linear in its deterministic integrand, whereas nonlinear preferences or probabilities require an explicit nonlinear mapping such as positive-part normalization (see Appendix~\ref{app:PositiveGaussianReweighting} for the Radon--Nikodym reweighting proposition and conditional half-normal simulations). Likewise, when full-line Hermite features $s_{1,j}(x)$ are restricted to a bounded outcome window $[0,L]$ via $\psi_{j,L}(t)=\sqrt{L}\,s_{1,j}(Lt)$ on $t\in[0,1]$, orthogonality is lost and simultaneous least-squares coefficients $c_L$ must be obtained from the Gram system $G_L c_L = b_L$ in \eqref{eq:BoundedGramProjection} rather than separate one-feature projections $q_j(L)=(b_L)_j/(G_L)_{jj}$. Appendix~\ref{app:BoundedWindowEstimation} reports numerical quadrature estimates and Gram condition numbers for CRRA utility $U(x)=x^{0.6}/0.6$ across $L\in\{1,2\}$, confirming accurate three-feature approximation on $[0,1]$ while illustrating how bounded-window collinearity affects separate versus joint coefficient magnitudes.
\subsection{Elicitation and testable restrictions}\label{subsec:ElicitationTestability}

\tab A finite elicitation exercise can be organized as follows. First, fix a truncation level $J$ and a basis ${s_{1j}}{j=1}^{J}$. The basis may be chosen ex ante, as in the Hermite approximation used below, or estimated from a training sample of payoff paths by an orthogonalization or principal-component step. Second, construct acts whose payoff paths differ only in known loadings on these basis features. For example, an act may be written as
\begin{equation}
	f_{\mathbf{a}}(x)=\sum_{j=1}^{J} a_j s_{1j}(x),
\end{equation}
where the loading vector $\mathbf{a}=(a_1,\ldots,a_J)^{\prime}$ is controlled by the experimenter. Third, elicit the local utility scale from simple constant or single-feature payoff paths, and then estimate the composite coefficients from certainty-equivalent, binary-choice, or pricing tasks using the projected utility representation
\begin{equation}
	U(f_{\mathbf{a}})\approx \sum_{j=1}^{J} u_j a_j .
\end{equation}
Fourth, recover normalized decision weights from repeated choices across acts with the same projected utility but different state-feature exposure. The design is identified when the matrix of observed loadings has full column rank; overidentifying restrictions are obtained by presenting more acts than the number of active feature coefficients.

\tab A natural stochastic-choice implementation is the Luce choice rule \citep{Luce1959}, or equivalently the softmax/logit rule used in modern stochastic-choice analysis \citep{CerreiaVioglioMaccheroniMarinacciRustichini2023}. Define the projected value
\begin{equation}
	V_J(f_{\mathbf{a}})=\sum_{j=1}^{J}p_j u_j a_j,
	\qquad
	p_j=\frac{v_j^+}{\sum_{k=1}^{J}v_k^+}.
\end{equation}
For any finite menu $B\subset\mathcal{A}_J$, the Luce probability of choosing act $f_{\mathbf{a}}\in B$ is
\begin{equation}
	\label{eq:LuceMenuChoice}
	\Pr(f_{\mathbf{a}}\mid B) =
\frac{\exp{\gamma V_J(f_{\mathbf{a}})}}
	{\sum_{f_{\mathbf{b}}\in B}\exp{\gamma V_J(f_{\mathbf{b}})}},
	\qquad \gamma>0.
\end{equation}
In the binary menu $B=\{f_{\mathbf{a}},f_{\mathbf{b}}\}$, this reduces to
\begin{equation}
	\label{eq:LuceBinaryChoice}
	\Pr(f_{\mathbf{a}}\succeq f_{\mathbf{b}}) =
	\Lambda\left(
	\gamma\sum_{j=1}^{J}p_j u_j(a_j-b_j)
	\right),
\end{equation}
where $\Lambda(z)=(1+\exp{-z})^{-1}$. Thus the binary logit specification is the Luce rule applied to projected abstract-Wiener acts, and $p_j$ remains the admissible nonnegative normalization of the positive part of the signed Wiener weight. Define the composite coefficient vector
\(\boldsymbol{\theta }=(\theta _{1},\dots ,\theta _{J})^{\prime },\qquad \theta _{j}=p_{j}u_{j},\)

so that $V_J(f_{\mathbf{a}})=\boldsymbol{\theta}^{\prime}\mathbf{a}$. For choice occasion $i$, let $B_i\subset\mathcal{A}_J$ denote the finite menu of projected acts offered to the decision maker, and let $f_{\mathbf{a}_i}\in B_i$ denote the observed chosen act. Thus a sample of $N$ menu choices is $\{(B_i,f_{\mathbf{a}_i})\}_{i=1}^{N}$, and its log-likelihood is
\begin{equation}
	\ell(\gamma,\boldsymbol{\theta}) =
	\sum_{i=1}^{N}\left[
	\gamma\boldsymbol{\theta}^{\prime}\mathbf{a}_i
	-\log\left(\sum_{f_{\mathbf{b}}\in B_i}\exp{\gamma\boldsymbol{\theta}^{\prime}\mathbf{b}}\right)
	\right].
\end{equation}
Standard errors can be obtained from the inverse observed information, by a sandwich covariance matrix, or by bootstrap when the basis is estimated in a first stage. Wald, likelihood-ratio, or minimum-distance tests then compare the unrestricted projected representation with SEU, RDU, CPT, or finite-state benchmarks.

\tab This procedure yields testable restrictions. The normalized weights must satisfy $p_j\geq0$ and $\sum_{j=1}^{J}p_j=1$. Monotonicity in monetary outcomes requires that, holding state-feature exposure fixed, higher payoff paths are weakly preferred. SEU is obtained only when the weights can be replaced by a state-independent probability vector that is separable from local utility. RDU and CPT impose additional rank-order and gain-loss restrictions on the normalized weights. The abstract Wiener model can therefore be rejected in a finite experiment if the estimated positive-part weights are unstable across equivalent acts, if monotonicity fails after normalization, if the implied weights violate the maintained rank or gain-loss ordering, or if additional Wiener basis features have no out-of-sample explanatory content relative to a finite-state benchmark.

\tab The empirical role of the raw signed weights is diagnostic rather than directly normative. A negative $v_j$ indicates a local signed deviation from the Gaussian reference measure in the $j$-th feature direction. It is not a negative probability and is not used as a choice probability. For admissible choice comparisons, the model uses $p_j$ or another specified nonnegative normalization. This is why the representation can accommodate signed Wiener integrals without violating the usual monotonicity requirement for preferences over monetary outcomes.
\begin{thm}[Finite projected representation and elicitation]\label{thm:FiniteProjectedElicitation}
	Fix a truncation level $J$ and linearly independent state features $\{s_{1j}\}_{j=1}^{J}$. Let the admissible act class be
	\(\mathcal{A}_{J}=\left\{f_{\mathbf{a}}(x)=\sum _{j=1}^{J}a_{j}s_{1j}(x):\mathbf{a}\in A\subset \mathbb{R}^{J}\right\},\)
	where $A$ is convex and has nonempty relative interior. Let $\succeq$ be a preference relation on $\mathcal{A}_J$. Suppose:
	\begin{enumerate}[label=(A\arabic*)]
		\item $\succeq$ is complete, transitive, and continuous in the loading vector $\mathbf{a}$;
		\item if $a_j\geq b_j$ for all coordinates that represent monetary payoff improvements and $\mathbf{a}\neq\mathbf{b}$, then $f_{\mathbf{a}}\succ f_{\mathbf{b}}$;
		\item for each feature $j$, there is a local utility scale $u_j$ such that tradeoffs across projected acts are separable in the feature loadings;
		\item admissible decision weights are normalized nonnegative transforms of the raw Wiener weights, so that $p_j\geq0$ and $\sum_{j=1}^{J}p_j=1$.
	\end{enumerate}
	Then there exist local utility coefficients $(u_1,\ldots,u_J)$ and normalized decision weights $(p_1,\ldots,p_J)$ such that preferences on $\mathcal{A}_J$ are represented by
	\(V_{J}(f_{\mathbf{a}})=\sum _{j=1}^{J}p_{j}u_{j}a_{j}=\boldsymbol{\theta }^{\prime }\mathbf{a},\qquad \boldsymbol{\theta }=(\theta _{1},\dots ,\theta _{J})^{\prime },\quad \theta _{j}=p_{j}u_{j}.\)
	Moreover, if the observed design matrix of feature loadings has rank $J$, then the composite coefficient vector $\boldsymbol{\theta}$ is identified from certainty-equivalent, pricing, or binary-choice observations on $\mathcal{A}_J$ up to the usual location and scale normalization of utility. If the local utility scale $(u_1,\ldots,u_J)$ is independently elicited or normalized with $u_j\neq0$ for all active coordinates, then the normalized decision weights are identified by
	\(p_{j}=\frac{\theta _{j}/u_{j}}{\sum _{k=1}^{J}\theta _{k}/u_{k}}.\)
	The finite projected representation is rejected on a design if the recovered weights fail nonnegativity, fail normalization, or imply a violation of monotonicity on the ordered act class. \qed
\end{thm}
\begin{proof}
	Each act in $\mathcal{A}_J$ is identified with its loading vector $\mathbf{a}\in A\subset\mathbb{R}^J$. Completeness, transitivity, and continuity give a continuous numerical representation on the finite projected act class. The separability condition restricts the representation to be additive across feature loadings after the local utility scales are fixed. Hence there exists a composite coefficient vector $\boldsymbol{\theta}=(\theta_1,\ldots,\theta_J)^{\prime}$ such that
	\(V_{J}(f_{\mathbf{a}})=\sum _{j=1}^{J}\theta _{j}a_{j}=\boldsymbol{\theta }^{\prime }\mathbf{a}.\)
	The normalization condition in (A4) decomposes each admissible composite coefficient as $\theta_j=p_ju_j$, where $p_j\geq0$ and $\sum_jp_j=1$. Monotonicity in (A2) rules out coefficient configurations that make a coordinatewise payoff improvement lower the represented value.
	For elicitation, let $M$ denote the matrix whose rows are the observed loading vectors. Certainty-equivalent or pricing observations identify the vector $\boldsymbol{\theta}$ from $M\boldsymbol{\theta}=\mathbf{y}$ whenever $\operatorname{rank}(M)=J$; in an overidentified design, the same coefficient vector is identified by the corresponding least-squares or minimum-distance normal equations. In binary-choice data, a standard single-index specification identifies $\boldsymbol{\theta}$ up to the scale of the link function. If $u_j$ is separately elicited or normalized and nonzero on the active coordinates, then $\theta_j/u_j$ identifies $p_j$ up to a common positive factor, and imposing $\sum_jp_j=1$ gives the displayed formula. Since raw Wiener weights are signed, nonnegativity, normalization, and monotonicity are not automatic mathematical consequences of the Wiener integral. They are admissibility restrictions on the decision-theoretic representation, so their empirical failure rejects the finite projected model on the maintained design.
\end{proof}
\begin{cor}[SEU, RDU, and CPT restrictions]\label{cor:SEURDUCPTRestrictions}
	Under \Cref{thm:FiniteProjectedElicitation}, the following restrictions are testable on the finite projected design.
	\begin{enumerate}[label=(\roman*)]
		\item SEU requires the normalized weights to be representable by a state-probability vector that is separable from local utility.
		\item RDU requires the normalized weights to be generated by an increasing weighting function applied to ordered cumulative probabilities.
		\item CPT requires a reference point and separate admissible gain and loss weights, with monotone cumulative decision weights on each side of the reference point.
	\end{enumerate}
	Failure of the relevant restriction rejects the corresponding SEU, RDU, or CPT submodel without necessarily rejecting the broader abstract-Wiener projected representation.
\end{cor}
\begin{proof}
	The theorem identifies the finite composite coefficients and, when the utility scale is fixed, the normalized decision weights. SEU, RDU, and CPT impose additional structure on those weights. SEU imposes separability of beliefs from local utility; RDU imposes rank-ordered cumulative weighting; and CPT adds reference dependence together with distinct gain and loss weighting functions. These restrictions are overidentifying restrictions on the same recovered finite vector of normalized weights. A violation therefore rejects the restricted submodel while leaving open the broader signed-Wiener representation with a different admissible normalization.
\end{proof}
\subsection{Luce menu-choice Monte Carlo}\label{subsec:JMELuceSimulation}

\tab This subsection conducts a repeated-sample Monte Carlo analysis of the Luce choice rule in \eqref{eq:LuceMenuChoice}--\eqref{eq:LuceBinaryChoice}. The exercise evaluates both identification and finite-sample recovery of the nonnegative normalization. Each menu contains three projected acts represented by four state-feature loadings on \(s_{1,1}\), \(s_{1,3}\), \(s_{1,5}\), and \(s_{1,7}\). The local utility scale is fixed at u=(1.10,0.85,0.55,0.30), while the raw Wiener-integral weights are allowed to be signed: v=(0.55,-1.15,0.25,0.40). The admissible weights are the positive-part normalized weights \(p_j=v_j^+ / \sum_k v_k^+\). This Monte Carlo conditions on the displayed fixed vector v; it is a coefficient-recovery exercise conditional on that realization, not an estimate of the ex ante expectation under \(P_{\star }\).

\tab Choices are generated from the Luce rule
\begin{equation}\label{eq:JMELuceSimProb} \Pr(f_{\mathbf{a}}\mid B) = \frac{\exp\{\gamma\boldsymbol{\theta}^{\prime}\mathbf{a}\}} {\sum_{f_{\mathbf{b}}\in B}\exp\{\gamma\boldsymbol{\theta}^{\prime}\mathbf{b}\}}, \qquad \theta_j=p_ju_j, \end{equation}
with \(\gamma\)=2.40. The unrestricted projected model estimates \(\boldsymbol{\theta }\) by maximizing the finite-menu likelihood. A simple finite-state benchmark restricts normalized feature weights to be equal and estimates only a common scale. For each of \(N\in\{1{,}800,5{,}000,10{,}000,50{,}000\}\) menus, we generate 1,000 independent datasets, re-estimate both models, and compute bias, root mean squared error (RMSE), Hessian-based 95\% coverage, sign recovery, classification accuracy, and the log-likelihood advantage per menu. The design therefore contains 4,000 estimations and approximately 66.8 million simulated menus.
\begin{table}[!htb!] \centering \caption{Repeated-sample performance of the Luce menu-choice estimator}\label{tab:JMELuceSummary} {\small
\begin{tabular}{rrrrrrrr}
\hline
$N$ & Reps. & Conv. & $|$Bias$|$ & RMSE & 95\% cov. & Hit rate & $\Delta LL/N$\\
\hline
 1,800 & 1000 & 1.0000 & 0.0312 & 0.0434 & 0.9463 & 0.5327 & 0.0493\\
 5,000 & 1000 & 1.0000 & 0.0181 & 0.0248 & 0.9505 & 0.5322 & 0.0486\\
10,000 & 1000 & 1.0000 & 0.0128 & 0.0176 & 0.9473 & 0.5323 & 0.0487\\
50,000 & 1000 & 1.0000 & 0.0057 & 0.0078 & 0.9473 & 0.5320 & 0.0484\\
\hline
\end{tabular}
} \par\vspace{0.35em} \begin{minipage}{0.90\textwidth} \footnotesize\textit{Notes:} Each row summarizes 1,000 independent Monte Carlo
		replications with three alternatives per menu. Bias and RMSE refer to the four
		normalized weights. Coverage averages the four Hessian-based 95\% intervals
		for the unrestricted choice coefficients. Hit rate is the unrestricted
		model's modal-choice accuracy. \(\Delta\) LL/N is its average log-likelihood
		advantage per menu over the re-estimated equal-weight benchmark.
\end{minipage} \end{table}
\begin{table}[!htb!] \centering \caption{Recovery of normalized decision weights at 50,000 menus}\label{tab:JMELuceWeights} {\small
\begin{tabular}{lrrrrrr}
\hline
Feature & Raw $v_j$ & True $p_j$ & Mean $\widehat p_j$ & RMSE($\widehat p_j$) & True $\theta_j$ & Mean $\widehat\theta_j$\\
\hline
$s_{1,1}$ & 0.5500 & 0.4583 & 0.4577 & 0.0081 & 0.5042 & 0.5040\\
$s_{1,3}$ & -1.1500 & 0.0000 & 0.0021 & 0.0038 & 0.0000 & -0.0001\\
$s_{1,5}$ & 0.2500 & 0.2083 & 0.2076 & 0.0075 & 0.1146 & 0.1143\\
$s_{1,7}$ & 0.4000 & 0.3333 & 0.3326 & 0.0103 & 0.1000 & 0.0999\\
\hline
\end{tabular}
} \par\vspace{0.35em} \begin{minipage}{0.90\textwidth} \footnotesize\textit{Notes:} Entries are averages across 1,000 independent
		replications with 50,000 three-alternative menus each. Raw coefficients
		\(v_{j}\) may be signed. True normalized weights satisfy
		\(p_j=v_j^+/\sum_k v_k^+\), and \(\theta_j=p_ju_j\). The negative raw weight
		on \(s_{1,3}\) therefore has a true admissible probability weight of zero.
\end{minipage} \end{table}

\tab \Cref{tab:JMELuceSummary} shows stable repeated-sample behavior. All 4,000 estimations converge. Average 95\% coverage ranges from 94.63\% to 95.05\%, close to its nominal level, and every replication recovers the signs of the three nonzero choice coefficients. Normalized-weight RMSE declines from 0.0434 at N=1,800 to 0.0078 at N=50,000. Across sample sizes, the projected model attains a modal hit rate near 53.2\%, compared with about 48.8\% for the equal-weight benchmark, and improves log likelihood by approximately 0.048--0.049 per menu.
\begin{figure}[!htb!] \centering \begin{subfigure}{0.48\textwidth} \centering \includegraphics[width=\linewidth]{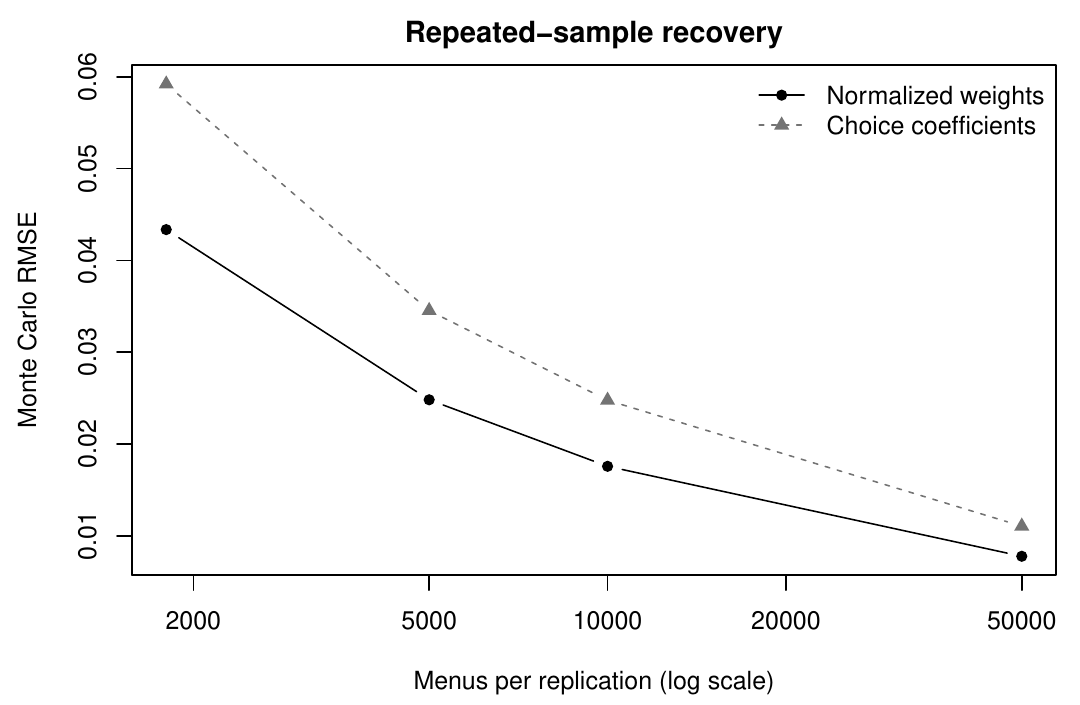} \caption{Monte Carlo RMSE by sample size.} \label{fig:JMELuceWeights} \end{subfigure}\hfill \begin{subfigure}{0.48\textwidth} \centering \includegraphics[width=\linewidth]{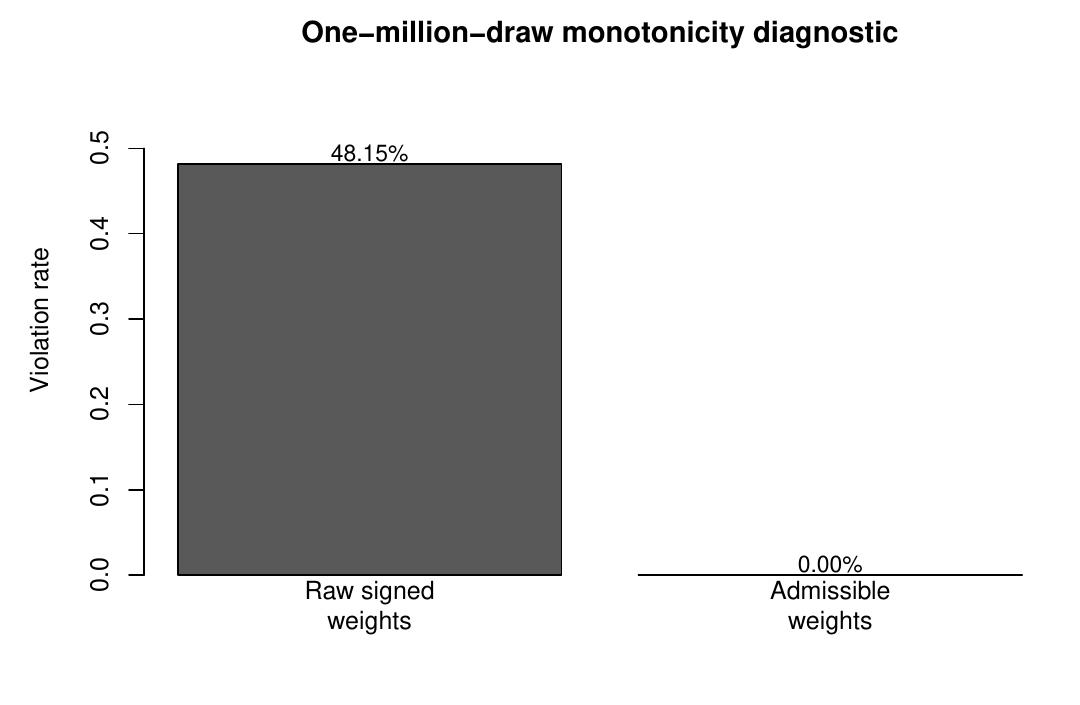} \caption{Monotonicity diagnostic.} \label{fig:JMELuceMonotonicity} \end{subfigure} \caption{Repeated-sample recovery and admissibility diagnostics.}\label{fig:JMELuceSimulation} \par\vspace{0.35em} \begin{minipage}{0.92\textwidth} \footnotesize\textit{Notes:} The left panel plots Monte Carlo RMSE for the
		normalized weights and choice coefficients across four sample sizes, using
		1,000 independent replications per size. The right panel reports monotonicity
		violations in one million independently generated coordinatewise payoff
		improvements.
\end{minipage} \end{figure}

\tab \Cref{tab:JMELuceWeights} and \cref{fig:JMELuceSimulation} show that the estimator recovers the economically relevant normalized weights rather than the raw signed coefficients. At N=50,000, the true weights and their Monte Carlo means are
\(p_{\mathrm{true}}\approx (0.4583,0,0.2083,0.3333),\qquad \mathbb{E}[\widehat{p}]\approx (0.4577,0.0021,0.2076,0.3326).\)
The one-million-draw diagnostic records a 48.15\% monotonicity-violation rate for raw signed weights and no violations for admissible weights. The small positive Monte Carlo mean for the second normalized weight reflects sampling error followed by the positive-part transformation; its true value is zero because the corresponding raw Wiener weight is negative. Thus raw Wiener-integral weights may be signed, but Luce choice probabilities and monotone preferences require admissible nonnegative decision weights.
\section{Conclusion}\label{sec:Conclusion}

\tab The paper's general results characterize preferences over path-valued acts and state when their coefficients can be identified from observed choices. These results do not depend on Gaussian uncertainty; the abstract Wiener model provides one application. Axioms~\ref{ax:HilbertWeakOrder}--\ref{ax:ConeMonotonicity} characterize a continuous linear projected-local representation over path-valued acts, normalization yields uniqueness (\Cref{thm:InfiniteProjectedRepresentation}), and an injective design information operator point-identifies the preference coefficient sequence from repeated stochastic choices (\Cref{thm:InfiniteIdentification}). Economically, finite menus recover an infinite-dimensional local preference functional whenever observed act differences span the maintained parameter space.

\tab In the abstract Wiener model, acts such as income, consumption, investment, or insurance-loss streams belong to a Banach space of paths. A Cameron–Martin Hilbert subspace provides local coordinates, and the Gaussian measure \(P_\star\) specifies the reference distribution for stochastic perturbations. Because raw Wiener-integral coordinates \(v_j(\omega)\) are centered and signed, economically admissible decision weights require an explicit nonnegative normalization such as \(p_j=v_j^+/\sum_k v_k^+\). Conditional on an elicited local utility scale, finite-dimensional projections recover these normalized weights and yield testable overidentifying restrictions for SEU, RDU, and CPT (\Cref{thm:FiniteProjectedElicitation} and \Cref{cor:SEURDUCPTRestrictions}).

\tab The repeated-sample Luce Monte Carlo study comprises 4,000 estimations across approximately 66.8 million simulated three-alternative menus. Across the maintained simulation designs, it delivers near-nominal coverage and declining coefficient-recovery error as sample size increases. A separate diagnostic involving one million coordinatewise payoff improvements records no monotonicity violations after admissible normalization. These simulation findings complement the representation's theoretical transitivity properties: under the common weights and deterministic utility scales in \eqref{eq:ExAnteRealizedOrders}, each fixed realization of \(\omega\) induces a transitive weak order, and the expected score induces a transitive ex ante order. Rankings may reverse across realizations without violating transitivity within any one realization. Formal extensions to It\^{o} and Clark--Ocone--Haussmann path-dependent utility representations are noted in Appendix~\ref{app:MalliavinExtension}.
\singlespacing
\noindent\(\textbf{Data\ and\ Code\ Availability:}\) No empirical datasets were generated or analyzed in this study. All numerical quadrature and Monte Carlo simulation scripts used to generate the figures and tables are available from the author upon request.

\FloatBarrier
\clearpage
\appendix
\onehalfspacing
\section*{Internet Appendix}
\phantomsection
\addcontentsline{toc}{section}{Internet Appendix}
\renewcommand{\thesection}{A.\arabic{section}}
\renewcommand{\theHsection}{appendix.\arabic{section}}
\renewcommand{\theequation}{\thesection.\arabic{equation}}
\setcounter{figure}{0}
\setcounter{table}{0}
\renewcommand{\thefigure}{A.\arabic{figure}}
\renewcommand{\thetable}{A.\arabic{table}}
\renewcommand{\theHfigure}{appendix.figure.\arabic{figure}}
\renewcommand{\theHtable}{appendix.table.\arabic{table}}
\counterwithin{cor}{section}
\counterwithin{prop}{section}
\counterwithin{lem}{section}
\counterwithin{assumption}{section}

\section{Cylindrical Set Probabilities and Odd-Feature Utility Normalization}
\label{app:CylindricalBornRule}
This section supplies the full-line normalization discussed in Section~2.2.1
of the main paper. Let $\{s_{1,j}\}_{j\geq0}$ be the real orthonormal
Hermite basis of $L^2(\mathbb R)$, and let the deterministic coefficient
sequence $\mathbf u=(u_0,u_1,\ldots)$ satisfy
$0<\sum_{j\geq0}u_j^2<\infty$. The series
$U(x;\mathbf u)=\sum_{j\geq0}u_js_{1,j}(x)$ then converges in
$L^2(\mathbb R)$. The Riesz--Fischer theorem gives
\citep[p.~70]{RieszNagy1955}
\begin{equation}\label{eq:AppParseval}
 \int_{-\infty}^{\infty}|U(x;\mathbf u)|^2\,dx
 =\sum_{j=0}^{\infty}u_j^2.
\end{equation}
Writing $\widehat u_j=u_j/\|\mathbf u\|_{\ell^2}$, where
$\|\mathbf u\|_{\ell^2}=(\sum_j u_j^2)^{1/2}$, yields
\begin{equation}\label{eq:AppUtilityNormalization}
 \left(\sum_{j=0}^{\infty}u_j^2\right)^{-1}
 \int_{-\infty}^{\infty}|U(x;\mathbf u)|^2\,dx
 =\int_{-\infty}^{\infty}|U(x;\widehat{\mathbf u})|^2\,dx=1.
\end{equation}
Thus $\rho_{\mathbf u}(x)=|U(x;\widehat{\mathbf u})|^2$ is a
probability density for an auxiliary outcome coordinate $X$. For a
Borel set $A\subset\mathbb R$, its pushforward through utility is
\begin{equation}\label{eq:AppUtilityPushforward}
 \Pr_{X\sim\rho_{\mathbf u}}\{U(X;\mathbf u)\in A\}
 =\int_{\{x:U(x;\mathbf u)\in A\}}\rho_{\mathbf u}(x)\,dx.
\end{equation}
Rescaling utility changes the numerical event unless the event set is
rescaled too; the same probability cannot automatically be assigned to
$U(X;\mathbf u)\in A$ and $U(X;\widehat{\mathbf u})\in A$.

\paragraph{What a cylindrical probability represents.}
A fixed coefficient vector is not a sequence of random events. To give
coefficients a probabilistic interpretation, specify a Borel probability
law $\nu$ on the sequence space $\ell^2$, the space of square-summable
real sequences. For coordinate sets $A_0,\ldots,A_J\subset\mathbb R$,
the cylinder probability is
\begin{equation}\label{eq:AppCylinderProbability}
 \nu\{u\in\ell^2:u_j\in A_j,\ 0\leq j\leq J\}.
\end{equation}
Its finite-dimensional marginals are consistent by construction. A
consistent family on the unrestricted product space need not be
supported on $\ell^2$; square summability must be checked separately.
For example, independent $u_j\sim N(0,\gamma_j)$ with
$\gamma_j>0$ and $\sum_j\gamma_j<\infty$ satisfy
$E\sum_j u_j^2=\sum_j\gamma_j<\infty$, hence belong to $\ell^2$
almost surely. Neither this law nor its cylinder probabilities follow
from the squared-amplitude density alone. Linking the two constructions
requires an additional model.

\paragraph{Odd-feature restriction.}
Setting $u_{2k}=0$ for $k\geq1$ removes nonconstant even Hermite
components. Since $s_{1,j}(-x)=(-1)^j s_{1,j}(x)$, the remaining
nonconstant odd part changes sign at the origin. This parity restriction
can produce shapes of interest for nonlinear utility
\citep{FriedmanSavage1948,Markowitz1952}; it does not by itself establish
monotonicity, particular curvature regions, or a preference ordering over
the features. Full-line orthogonality must not be carried over to
bounded restrictions without a Gram-matrix adjustment, as in
\Cref{app:BoundedWindowEstimation}.

\section{Decision-Weight Fluctuations Across Independent Repetitions}
\label{app:IndependentWeightLIL}
Fix one component of a normalized finite decision-weight vector, and let
$\pi^{(k)}\in[0,1]$ be its value in repetition $k$. Assume repetitions
are independent and identically distributed. Define
$m=E[\pi^{(1)}]$, $\sigma^2=\operatorname{Var}(\pi^{(1)})>0$, and
$S_n=\sum_{k=1}^n(\pi^{(k)}-m)$. The classical law of the iterated
logarithm (LIL) gives \citep[\S10.2]{ChowTiecher1988}
\begin{equation}\label{eq:AppScalarLIL}
 \limsup_{n\to\infty}\frac{S_n}{\sqrt{2n\sigma^2\log\log n}}=1,
 \qquad
 \liminf_{n\to\infty}\frac{S_n}{\sqrt{2n\sigma^2\log\log n}}=-1
 \quad\text{almost surely}.
\end{equation}
Consequently the centered sample mean has almost-sure fluctuation scale
$\sqrt{2\sigma^2\log\log n/n}$. A zero-variance component instead has
zero centered fluctuations. This statement concerns a single component
across repetitions, not independent components within a weight vector:
normalization couples those components because their sum is one.
More general Banach-valued results require their additional covariance
and integrability hypotheses \citep{Kuelbs1977}.

\section{Positive Reweighting and Conditional Gaussian Diagnostics}
\label{app:PositiveGaussianReweighting}
Let $(\Omega,\mathcal F,P_\star)$ be the Gaussian reference probability
space. Let $\xi_j:\Omega\to\mathbb R$ be Borel measurable with
$\int|\xi_j|\,dP_\star<\infty$, and define
$\xi_j^+=\max\{\xi_j,0\}$. The signed density defines
$\widetilde Q_j(A)=\int_A\xi_j\,dP_\star$. If
$\int\xi_j^+\,dP_\star>0$, its normalized positive part is
\begin{equation}\label{eq:AppSignedChangeOfMeasure}
 \frac{dQ_j^+}{dP_\star}
 =\frac{\xi_j^+}{\int\xi_j^+\,dP_\star}.
\end{equation}
\begin{prop}[Positive reweighting and conditional Gaussian weights]
\label{prop:AppNonlinearQ}
The displayed density defines a probability measure. Separately, if
$v_j\sim N(0,\sigma_j^2)$ with $\sigma_j>0$, its conditional density
given $v_j>0$ is
\[
 f_j^+(z)=\frac{2}{\sigma_j\sqrt{2\pi}}
 \exp\!\left(-\frac{z^2}{2\sigma_j^2}\right)
 \mathbf 1_{\{z>0\}}.
\]
The clipped variable $v_j^+=\max\{v_j,0\}$ instead has mass $1/2$
at zero and is not half-normal. A signed density, a Gaussian Wiener
coordinate, and its clipped or conditioned version are different objects.
\end{prop}
\begin{proof}
The Radon--Nikodym density is nonnegative and integrates to one. Gaussian
symmetry gives $\Pr(v_j>0)=1/2$, so conditioning doubles the Gaussian
density on $(0,\infty)$. Clipping retains the probability of
$v_j\leq0$ at zero instead. The density of a reweighted path law
$Q_j^+$ need not equal this conditional Gaussian density.
\end{proof}
\begin{figure}[!htb!]
\centering
\begin{minipage}[h]{0.4\linewidth}
	\captionof{figure}{Conditional positive Gaussian draws}
	\label{fig:AppQsimU11}
	\centerline{\includegraphics[width=\linewidth]{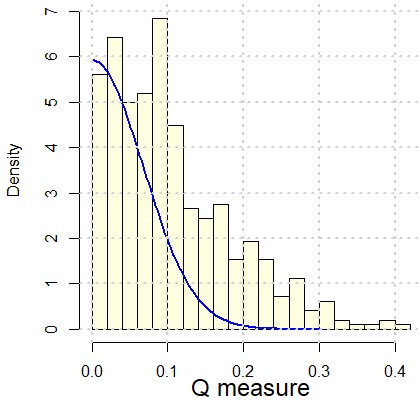}}
\end{minipage}
\hspace*{1.8cm}
\begin{minipage}[h]{0.4\linewidth}
	\captionof{figure}{Ordered conditional Gaussian draws}
	\label{fig:AppQmeasU11}
	\vspace*{0.25cm}\centerline{\includegraphics[width=\linewidth]{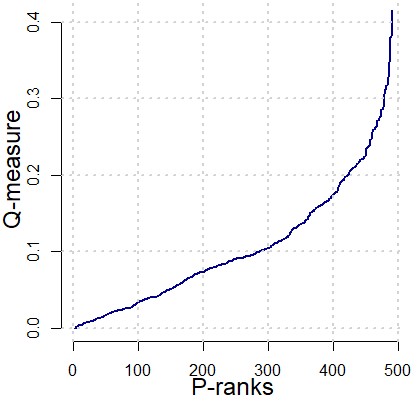}}
\end{minipage}
\begin{minipage}[h]{0.9\textwidth}
	\vspace*{0.5cm}
	\footnotesize\textit{Notes:} These legacy simulation panels illustrate
1,000 draws from a centered Gaussian coordinate conditional on positivity.
The left panel displays their distribution; the right panel displays ordered
draws. They illustrate conditional Gaussian realizations, not probabilities
across states or the Radon--Nikodym density of $Q_j^+$. The clipped positive
part would additionally retain an atom at zero, as explained in
\Cref{prop:AppNonlinearQ}.
\end{minipage}
\end{figure}
\FloatBarrier
\section{Bounded-Window Hermite Quadrature and Gram Conditioning}
\label{app:BoundedWindowEstimation}
To illustrate the distinction between full-line orthogonal Fourier projections and bounded-window least-squares projections, consider the CRRA utility function $U(x)=x^{0.6}/0.6$ using physicists' Hermite polynomials $H_1(x)=2x$, $H_3(x)=8x^3-12x$, and $H_5(x)=32x^5-160x^3+120x$, with full-line normalization $A_j=(2^j j!\sqrt\pi)^{-1/2}$. For each window endpoint $L\in\{1,2\}$, all integrals are evaluated on the canonical domain $t\in[0,1]$ using rescaled features $\psi_{j,L}(t)=\sqrt{L}\,s_{1,j}(Lt)$ and target $\sqrt{L}\,U(Lt)$. \Cref{tab:AppCanonicalHermiteCoefficients} reports both the separate one-feature coefficients $q_j(L)=(b_L)_j/(G_L)_{jj}$ and the simultaneous Gram least-squares coefficients $c_j(L)$ solving $G_L c_L = b_L$.
\begin{figure}[!htb!]
\centering
\begin{minipage}[h]{0.4\linewidth}
	\captionof{figure}{Bounded-window feature loadings}
	\label{fig:AppAbstWherm3int0}
	\centerline{\includegraphics[width=\linewidth]{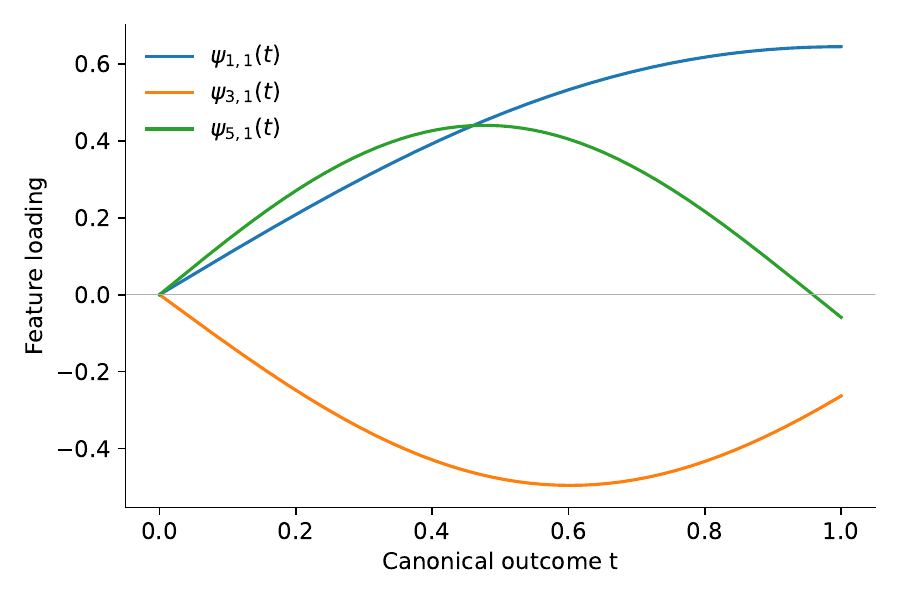}}
\end{minipage}
\hspace*{1.8cm}
\begin{minipage}[h]{0.4\linewidth}
	\captionof{figure}{Joint approximation of $U_1(t)$}
	\label{fig:AppAbstWherm3CRRA4}
	\vspace*{0.25cm}\centerline{\includegraphics[width=\linewidth]{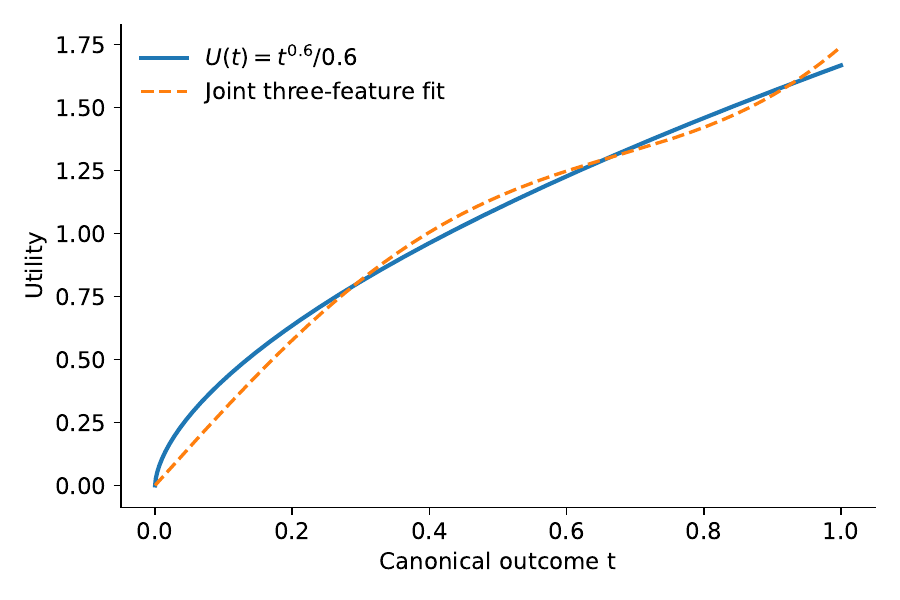}}
\end{minipage}
\begin{minipage}[h]{0.9\textwidth}
	\vspace*{0.5cm}
	\footnotesize\textit{Notes:} The horizontal coordinate is the canonical outcome $t\in[0,1]$ with $L=1$. The left panel plots the restricted feature loadings $\psi_{1,1}(t)$, $\psi_{3,1}(t)$, and $\psi_{5,1}(t)$. The right panel plots $U(t)=t^{0.6}/0.6$ alongside its simultaneous three-feature Gram least-squares approximation.
\end{minipage}
\end{figure}
\begin{table}[!htb!]\centering\small
\caption{Bounded-window Hermite coefficients on the canonical domain}
\label{tab:AppCanonicalHermiteCoefficients}

\begin{tabular}{rrrr}
\hline
Window endpoint $L$ & Feature $j$ & Separate $q_j(L)$ & Joint $c_j(L)$\\
\hline
1 & 1 & 2.423501 & 7.052685 \\
1 & 3 & -2.844874 & 9.452479 \\
1 & 5 & 2.992357 & 5.385016 \\
2 & 1 & 3.249333 & 3.724097 \\
2 & 3 & 0.724925 & 2.483367 \\
2 & 5 & -1.828388 & 1.299431 \\
\hline
\end{tabular}

\end{table}
\FloatBarrier
Changing $L$ alters both the outcome window and the restricted feature Gram matrix $G_L$. Cancellations in $\int U s_{1,j}$ can reverse the sign of a separate projection coefficient $q_j(L)$, while near collinearity on narrow windows amplifies simultaneous coefficients $c_j(L)$. Consequently, individual coefficient magnitudes cannot be interpreted in isolation without accounting for feature norms, Gram covariance, and approximation error.
For $L=1$, the Gram condition number is 15666.51 and the joint-fit
$L^2[0,L]$ error is 0.054896. For $L=2$, the corresponding values are
13.85 and 0.080166. Increasing Gauss--Legendre quadrature from 512 to
1,024 nodes changes each joint coefficient by less than
$3.5\times10^{-11}$.

\section{Dynamic Path-Dependent Extension via Malliavin Calculus}
\label{app:MalliavinExtension}
Fix a finite horizon $T>0$, and let $\{W_t\}_{0\leq t\leq T}$ be
a standard Brownian motion on $(\Omega,\mathcal F,P)$ with its completed
natural filtration $\{\mathcal F_t\}_{0\leq t\leq T}$. A filtration is
the increasing collection of information sets available over time.
Let $Z=\mathcal V(W_{[0,T]})$ be a square-integrable,
$\mathcal F_T$-measurable terminal utility functional of the path.
The Brownian martingale representation theorem gives a predictable
integrand $\varphi$ with $E\int_0^T|\varphi_t|^2dt<\infty$ such that
\citep{Oksendal1997}
\begin{equation}\label{eq:AppMartingaleRepresentation}
 Z=E[Z]+\int_0^T\varphi_t\,dW_t,\qquad
 E[Z\mid\mathcal F_t]=E[Z]+\int_0^t\varphi_s\,dW_s.
\end{equation}
Predictability is the measurability condition ensuring that the integrand
does not anticipate Brownian innovations. The space of such integrands
is denoted $L^2_{\mathrm{pred}}(\Omega\times[0,T])$; it is distinct
from the deterministic Cameron--Martin path space $H$ in the main paper.
The representation need not hold with respect to a single Brownian
motion in an arbitrary enlarged filtration.

If $Z$ also belongs to the Malliavin Sobolev space $\mathbb D^{1,2}$,
meaning that $Z$ and its Malliavin derivative $D_tZ$ have finite second
moments in the corresponding norms, the Clark--Ocone formula identifies
a version of the integrand:
\begin{equation}\label{eq:AppClarkOcone}
 \varphi_t=E[D_tZ\mid\mathcal F_t]
 \quad\text{for almost every }t\text{ and }P\text{-almost surely}.
\end{equation}

For an $n$-dimensional Brownian motion with its natural filtration,
the integral in \eqref{eq:AppMartingaleRepresentation} becomes
$\sum_{k=1}^n\int_0^T\varphi_t^k\,dW_t^k$, and
$\varphi_t^k=E[D_t^kZ\mid\mathcal F_t]$ under the analogous
Malliavin differentiability condition
\citep{Nualart2006,Mataramvura2012}.
These are conditional representation results for square-integrable
random utility functionals. They do not by themselves establish the
economic axioms, identification, or admissibility of random decision
weights; those extensions remain future work.

\FloatBarrier
\clearpage
\singlespacing
\phantomsection
\addcontentsline{toc}{section}{References}
\bibliographystyle{chicago}
\bibliography{RankDepRiskEst}
\end{document}